\documentclass[a4paper,UKenglish,cleveref, autoref, thm-restate]{lipics-v2021}

\usepackage{amsmath}
\usepackage{mathtools}
\usepackage{tikz}
\usetikzlibrary{arrows.meta,positioning}
\newcommand{\ClaimCmd}{\mathsf{claim}}
\newcommand{\ReserveCmd}{\mathsf{reserve}}
\newcommand{\UseCmd}{\mathsf{use}}
\newcommand{\FreeState}{\mathsf{free}}
\newcommand{\UsedState}{\mathsf{used}}
\newcommand{\ReservedState}{\mathsf{reserved}}
\newcommand{\VAuth}{\mathsf{VAuth}}

\title{Brief Announcement: Fair Binding for Hidden-State Authorization in Byzantine SMR}

\titlerunning{Fair Binding for Hidden-State Authorization}

\author{Arnab Mallick}{Centre for Development of Advanced Computing, Hyderabad, India \and \url{https://arnab-m1.github.io/}}{arnabm@cdac.in}{https://orcid.org/0009-0008-5458-6429}{}

\author{Indraveni Chebolu}{Centre for Development of Advanced Computing, Hyderabad, India \and \url{}}{indravenik@cdac.in}{https://orcid.org/0009-0009-6128-6253}{}
{}

\authorrunning{Arnab Mallick}

\Copyright{Arnab Mallick}

\ccsdesc[500]{Computing methodologies~Distributed algorithms}
\ccsdesc[500]{Computer systems organization~Dependable and fault-tolerant systems and networks}
\ccsdesc[500]{Security and privacy~Distributed systems security}

\keywords{Byzantine consensus, validated agreement, order fairness, hidden state, authorization}

\category{Brief Announcement} 

\supplementdetails[linktext={GitHub repository},subcategory={Lean 4 formalization}]
{Software}
{https://github.com/Arnab-m1/Fair-Binding-for-Hidden-State-Authorization}

\nolinenumbers

\EventEditors{Ioannis Chatzigiannakis, Andrea Vitaletti, Keren Censor-Hillel, and William K. Moses Jr.}
\EventNoEds{4}
\EventLongTitle{40th International Symposium on Distributed Computing (DISC 2026)}
\EventShortTitle{DISC 2026}
\EventAcronym{DISC}
\EventYear{2026}
\EventDate{November 9--13, 2026}
\EventLocation{Rome, Italy}
\EventLogo{}
\SeriesVolume{397}
\ArticleNo{58}

\begin{document}

\maketitle

\begin{abstract}
Validated Byzantine SMR assumes that replicas can evaluate the validity of an ordered command. Agent authorization creates a different regime: a command may be valid only relative to a committed policy state that validators cannot reconstruct from the log. A proof that an action was authorized at an old commitment is then only a historical attestation, it does not by itself reserve the hidden resource for later use.

We isolate two independent requirements for safe live allocation of a hidden consumable resource under a Byzantine leader. First, arrival order at correct replicas must constrain commit order, the gap addressed by fair-ordering protocols. Second, a committed first request must bind later validity: it must make conflicting later requests invalid, not merely record that the first request was once authorized. The second requirement is non-vacuous precisely because the current policy state is hidden and not prefix-recoverable. Using an explicit authorization-witness interface, we characterize the two distinct obligations in this one-shot reservation model and give a fair reserve/use protocol satisfying both authorization safety and first-arrival liveness. Under trusted FIFO admission the two requirements collapse because admission and execution are atomic, Byzantine SMR separates request commitment from use.
\end{abstract}

\section{Introduction}

Byzantine state machine replication (SMR)~\cite{lamport1982byzantine,castro1999pbft} makes correct replicas agree on an ordered log. Validated Byzantine SMR~\cite{cachin2001broadcast} further requires each ordered value to satisfy a validator-checkable external-validity predicate. Standard applications make the predicate public or derivable from replicated state. Agent systems create a less classical case: authorization can depend on a policy state tracking budgets, permissions, retrieved context, or private balances that is committed but not replicated at validators. Validators see authenticated commitments to policy state and proof-carrying evidence, but cannot reconstruct the current state from the log prefix.

This distinction matters for contended consumable resources. Suppose a hidden policy state contains one available unit of a resource. Two clients may both receive valid evidence that their action was authorized at the same commitment $c_i$. A Byzantine leader can commit the second-arriving reservation before the first. Fair ordering addresses that arrival-order gap. But even if the first client's \emph{claim} is fairly ordered first, a claim that merely attests ``the resource was free at $c_i$'' does not reserve the resource: a Byzantine leader can order a conflicting reservation and use between the claim and the first client's later use. Figure~\ref{fig:separation} summarizes the two failure modes and their composition.

\paragraph*{One-shot example.}
A confidential policy permits one remaining API invocation. Agents $a$ and $b$ obtain valid witnesses against the same commitment while the unit is available. If $a$'s historical claim commits without creating a public reservation, validators still have no prefix-visible reason to reject $b$'s later reservation. Committing $\ReservedState(r,a)$, in contrast, makes $b$'s conflicting reservation invalid.

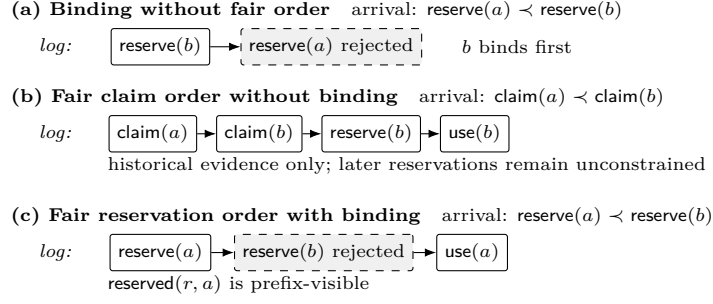
\begin{figure}[t]
  \centering
  \begin{tikzpicture}[
      entry/.style={draw,rounded corners=1pt,minimum height=4.5mm,
        inner xsep=3pt,font=\scriptsize},
      rejected/.style={entry,dashed,fill=black!6},
      flow/.style={-{Latex[length=1.5mm]},thin},
      heading/.style={font=\scriptsize,anchor=west},
      loglabel/.style={font=\scriptsize\itshape,anchor=east}
    ]
    \node[heading] at (0,0)
      {\textbf{(a) Binding without fair order}\quad arrival:
       $\ReserveCmd(a)\prec\ReserveCmd(b)$};
    \node[loglabel] at (1.05,-0.48) {log:};
    \node[entry] (a-b) at (2.1,-0.48) {$\ReserveCmd(b)$};
    \node[rejected,right=4mm of a-b] (a-a) {$\ReserveCmd(a)$ rejected};
    \draw[flow] (a-b) -- (a-a);
    \node[font=\scriptsize,anchor=west,right=4mm of a-a] {$b$ binds first};

    \node[heading] at (0,-1.15)
      {\textbf{(b) Fair claim order without binding}\quad arrival:
       $\ClaimCmd(a)\prec\ClaimCmd(b)$};
    \node[loglabel] at (1.05,-1.63) {log:};
    \node[entry] (b-ca) at (2.0,-1.63) {$\ClaimCmd(a)$};
    \node[entry,right=2.5mm of b-ca] (b-cb) {$\ClaimCmd(b)$};
    \node[entry,right=2.5mm of b-cb] (b-rb) {$\ReserveCmd(b)$};
    \node[entry,right=2.5mm of b-rb] (b-ub) {$\UseCmd(b)$};
    \draw[flow] (b-ca) -- (b-cb);
    \draw[flow] (b-cb) -- (b-rb);
    \draw[flow] (b-rb) -- (b-ub);
    \node[font=\scriptsize,anchor=west] at (1.28,-2.05)
      {historical evidence only; later reservations remain unconstrained};

    \node[heading] at (0,-2.72)
      {\textbf{(c) Fair reservation order with binding}\quad arrival:
       $\ReserveCmd(a)\prec\ReserveCmd(b)$};
    \node[loglabel] at (1.05,-3.20) {log:};
    \node[entry] (c-ra) at (2.1,-3.20) {$\ReserveCmd(a)$};
    \node[rejected,right=3mm of c-ra] (c-rb) {$\ReserveCmd(b)$ rejected};
    \node[entry,right=3mm of c-rb] (c-ua) {$\UseCmd(a)$};
    \draw[flow] (c-ra) -- (c-rb);
    \draw[flow] (c-rb) -- (c-ua);
    \node[font=\scriptsize,anchor=west] at (1.28,-3.62)
      {$\ReservedState(r,a)$ is prefix-visible};
  \end{tikzpicture}
  \caption{\textbf{Two distinct obligations and their composition.}
  Fair ordering protects a separated first arrival at the priority-command
  phase, binding determines what the winning committed command makes invalid
  afterward. Fairly ordered historical claims do not constrain later
  reservations, whereas a fairly ordered binding reservation does.}
  \label{fig:separation}
\end{figure}

\paragraph*{Contribution.}
We formalize this as \emph{hidden-state authorization}: validators can verify witnesses against commitments but cannot compute current authorization from the committed prefix. For a first-arriving honest request to safely and eventually execute, Byzantine SMR needs two orthogonal properties:
\begin{itemize}
  \item \textbf{Fair commit order.} If all correct replicas receive request $a$ before any receive conflicting request $b$, then $a$ must commit before $b$.
  \item \textbf{Binding validity.} Once $a$ commits, the validator must reject later conflicting requests. A historical attestation is not enough, the log must contain a reservation or an equivalent prefix-readable validity rule.
\end{itemize}

We make the witness interface explicit. We show that reservation priority requires fair commit order and that fair ordering of nonbinding historical claims alone is insufficient. We then give a two-phase reserve/use construction. The characterization is for one indivisible, non-replenishing resource, replenishing and multi-resource conflicts are outside its scope.

\section{Model}
\label{sec:model}

\paragraph*{Validated Byzantine SMR.}
$n=3f+1$ replicas, at most $f$ of which may be Byzantine and deviate arbitrarily, maintain an ordered log. The protocol is leader-based and partially synchronous: after an unknown GST, message delays are bounded and ordinary BFT liveness applies, timeouts trigger leader rotation. A \emph{slot} is a position in the committed log. Before ordering a command, correct replicas evaluate its external-validity predicate.

\begin{definition}[One-resource hidden-state authorization object]
  \label{def:object}
  Let $\mathcal S,\mathcal C,\mathcal R,\mathcal A,\mathcal E$ be the domains of policy states, commitments, resources, actions, and witnesses. The policy predicate and public witness verifier have signatures
  \[
    P:\mathcal S\!\times\!\mathcal R\!\times\!\mathcal A\to\{0,1\},
    \qquad
    \VAuth:\mathcal C\!\times\!\mathcal R\!\times\!\mathcal A\!\times\!\mathcal E\to\{0,1\}.
  \]
  The object exports $\ReserveCmd(r,a,e)\to\{\mathsf{commit},\mathsf{reject}\}$ and $\UseCmd(r,a)\to\{\mathsf{execute},\mathsf{reject}\}$. We use the auxiliary historical command $\ClaimCmd(c,r,a,e)\to\{\mathsf{commit},\mathsf{reject}\}$ to separate attestation from reservation. A matching reservation and use carry the same pair $(r,a)$.
\end{definition}

\paragraph*{Notation and timing.}
Let $L_{<j}$ be the committed prefix before slot $j$. Validation for slot $j$ uses $L_{<j}$ and the current policy commitment $c_j$. Validators maintain a set $K_j$ of authenticated commitments known from the prefix or authenticated policy-state interface. A prefix contains a \emph{matching reservation} for $a$ if it contains $\ReservedState(r,a)$, and a \emph{conflicting reservation} if it contains $\ReservedState(r,x)$ for some $x\#_r a$. A witness $e$ is \emph{current-valid} for $(r,a)$ at slot $j$ when $\VAuth(c_j,r,a,e)=1$.

\paragraph*{Hidden policy state.}
At log height $j$, the policy state is $S_j$ and its authenticated commitment is $c_j$. Validators know $c_j$ but do not know $S_j$ and cannot reconstruct it from $L_{<j}$. All correct validators evaluate authorization against the same authenticated current commitment: the policy-state interface does not equivocate among them. If validators could compute $S_j$ from the prefix, ordinary revalidation would subsume binding.

\paragraph*{Authorization witnesses.}
Soundness says that if $c$ binds state $S$ and $\VAuth(c,r,a,e)=1$, then $P(S,r,a)=1$; completeness gives an accepting witness for every authorized action. An opening, zero-knowledge proof, or transition certificate is an example realization of $e$, not a model primitive. Evidence is commitment-relative: evidence for $c_i$ proves authorization at $c_i$, not continued validity after later transitions.

\paragraph*{Hidden consumable resource.}
The theorem concerns one indivisible, non-replenishing resource $r$ whose hidden projection is in $\{\FreeState,\UsedState\}$. A consuming action is enabled only at $\FreeState$ and changes that projection to $\UsedState$; no transition restores it. Write $a\#_r b$ when both compete for that unit, so no legal execution from the same available-resource state can execute both. This is the one-shot incompatibility specialization of state-dependent conflict and non-commutativity~\cite{kuznetsov2025conflicts}. Replenishing balances and multi-resource transactions are outside the characterization.

\paragraph*{Command classes.}
A historical command $\ClaimCmd(c,r,a,e)$ is valid at slot $j$ when $c\in K_j$ and $\VAuth(c,r,a,e)=1$; committing it adds no public conflict fact and does not advance the hidden commitment. A reservation $\ReserveCmd(r,a,e)$ is valid when $e$ is current-valid and $L_{<j}$ contains no conflicting reservation, on commit it appends $\ReservedState(r,a)$. A use $\UseCmd(r,a)$ is valid only after a matching reservation, successful use consumes the hidden unit. Reservations are irrevocable: there is no release, cancellation, or expiry.

For a protocol, its \emph{priority command} is the command whose ordering establishes first-arrival precedence. It is $\ReserveCmd$ in the construction and $\ClaimCmd$ in the claim-priority baseline, that baseline still requires a later reservation before use.

\begin{definition}[Fair commit order and binding validity]
  \label{def:fairbinding}
  \emph{Fair commit order} requires that if every correct replica receives the priority command for $a$ before any correct replica receives a conflicting priority command for $b$, then $b$ does not commit before $a$. \emph{Binding validity} requires that once the winning priority command for $a$ commits, the prefix contains a public conflict fact that makes every later conflicting reservation invalid. In the reserve/use construction this fact is $\ReservedState(r,a)$.
\end{definition}

\begin{definition}[Evidence-local validation]
  \label{def:evloc}
  A validator is \emph{evidence-local} if its decision at slot $j$ may use $L_{<j}$, $K_j$, the current commitment $c_j$, public facts $\ReservedState(\cdot,\cdot)$ recorded in the prefix, syntactic well-formedness, and $\VAuth$ on supplied witnesses. It does not query or reconstruct $S_j$. Thus, unless the prefix contains a binding fact, authorization evidence does not make conflicting later commands invalid.
\end{definition}

\begin{definition}[Authorization safety and first-arrival liveness]
  \label{def:properties}
  \emph{Authorization safety} requires that no two actions $a\#_r b$ both execute and every executed $\UseCmd(r,a)$ has a prior matching reservation. \emph{First-arrival liveness} requires that if the honest priority command for $a$ is received by every correct replica before any correct replica receives a conflicting priority command, then $a$ eventually executes under the usual post-GST conditions, assuming the honest client remains live, supplies current-valid evidence until its reservation is ordered (refreshing it if necessary), and submits the required follow-up reservation and use. For interleaved conflicting arrivals no first-arrival guarantee is imposed, safety still permits at most one conflicting use.
\end{definition}

\begin{proposition}[Trusted-FIFO baseline]
  \label{prop:fifo}
  Under trusted FIFO admission, the first request accepted by the trusted path is ordered and executed atomically before any later conflict. Fair ordering and consensus-resident binding are therefore not separate obligations. Byzantine-leader SMR separates request commitment from later use, exposing both gaps.
\end{proposition}

\section{Necessity}

We separate two requirements. Fair ordering is about \emph{which} request reaches the log first, binding validity is about \emph{what} the first log entry makes invalid later.

\begin{lemma}[Fair order is necessary for first-arrival liveness]
  \label{lem:fairnecessary}
  Consider an authorization-safe protocol using the reservation rule above. If it permits a later conflicting reservation to commit before a first-arriving honest reservation, then it can violate first-arrival liveness.
\end{lemma}

\begin{proof}
  Let honest $\ReserveCmd(r,a,e_a)$ be received by all correct replicas before any correct replica receives $\ReserveCmd(r,b,e_b)$ with $a\#_r b$, and suppose both witnesses are current-valid at the relevant slot. By the negative assumption, some Byzantine-leader execution commits $b$ first. The prefix now contains $\ReservedState(r,b)$, a conflicting reservation for $a$. To preserve authorization safety, the protocol must reject or fail to execute $a$; accepting both would allow conflicting uses of $r$. Hence the first-arriving honest request is not live. First-arrival liveness therefore requires fair commit order for conflicting reservations.
\end{proof}

This is exactly the gap targeted by order-fairness mechanisms~\cite{kelkar2020order,kelkar2023themis,zhang2020pompe}.

\begin{lemma}[Historical attestations are insufficient]
  \label{lem:bindingnecessary}
  Suppose first-arrival precedence is enforced only by fairly ordering historical $\ClaimCmd$ commands, while committed claims neither create a public conflict fact nor constrain the order of later conflicting reservations. Then fair claim order alone does not imply authorization safety plus first-arrival liveness.
\end{lemma}

\begin{proof}
  Suppose fair ordering guarantees that an attestation $\ClaimCmd(c_i,r,a,e_a)$ commits before a conflicting attestation $\ClaimCmd(c_i,r,b,e_b)$, where $a\#_r b$. By definition, an attestation-only claim records $\VAuth(c_i,r,a,e_a)=1$ but appends no $\ReservedState(r,a)$, no other conflict fact, and no update that makes the hidden commitment unavailable for later conflicting authorization. Consider an execution in which no authenticated policy-commitment update occurs between these slots, so the current commitment remains $c_i$. Immediately after the committed claim for $a$, the public prefix still contains no fact making $b$ invalid, and $e_b$ can be current-valid for $b$. A Byzantine leader can therefore propose $\ReserveCmd(r,b,e_b)$ before $a$'s later use. If $b$ is accepted, $a$ can no longer execute without double-spending $r$; if $a$ is later also accepted, safety is violated. Rejecting $a$ preserves safety but violates first-arrival liveness. Fair ordering of attestations therefore does not suffice, the first committed claim must bind subsequent validity.
\end{proof}

\begin{observation}[Orthogonality]
  Fair order and binding validity are independent. A binding reservation without fair order can bind the wrong, later-arriving request first. A fair-ordered attestation without binding leaves later conflicting validity unchanged. The two obligations therefore address distinct failure modes in the one-shot reservation model.
\end{observation}

\section{Sufficiency Construction}

\begin{proposition}[Fair reserve/use]
  \label{prop:reserveuse}
  Assume a reservation-ordering layer satisfying fair commit order (Definition~\ref{def:fairbinding}). The reserve/use protocol satisfies authorization safety and first-arrival liveness for the one-shot hidden resource.
\end{proposition}

\begin{proof}[Construction]
  \emph{Reserve.} At slot $j$, a client submits a reservation for $(r,a)$ with evidence $e$. Validity requires
  \[
    \VAuth(c_j,r,a,e)=1
  \]
  and no conflicting reservation in $L_{<j}$. On commit, the log records $\ReservedState(r,a)$ as a public binding fact.

  \emph{Use.} A client submits $\UseCmd(r,a)$. Validity requires a prior $\ReservedState(r,a)$ in the prefix. No conflicting $\UseCmd(r,b)$ is valid because the prefix test prevented a conflicting reservation and use requires a matching one.

  \emph{Safety.} By induction on slots, once $\ReservedState(r,a)$ appears, every later $\ReserveCmd(r,b,\cdot)$ with $a\#_r b$ is invalid, and every use names a committed reservation. Hence conflicting uses cannot both execute. \emph{Liveness.} If honest $\ReserveCmd(r,a,\cdot)$ is received by all correct replicas before any conflicting reservation, fair ordering puts it first while the client's supplied or refreshed evidence remains current-valid. The reservation binds the prefix, so later conflicts are invalid. After GST, ordinary BFT liveness orders the matching use.
\end{proof}

\paragraph*{Algorithmic interface.}
The construction treats fair ordering for reservation commands as a service supplied by the consensus layer. Binding is implemented by the application's external-validity predicate: once $\ReservedState(r,a)$ appears in the committed prefix, validators reject every conflicting reservation. The two obligations therefore compose at different layers-ordering in consensus and conflict exclusion in application validity.

\begin{theorem}[One-resource characterization]
  \label{thm:characterization}
  For evidence-local Byzantine SMR implementing Definition~\ref{def:object} for one indivisible, non-replenishing resource under the reservation-admissibility rule above: 
  \begin{itemize}
      \item an authorization-safe implementation satisfying first-arrival liveness for reservation-priority requests requires fair commit order for conflicting reservations;
      \item in the claim-priority baseline, fair ordering of historical claims without a prefix-visible binding fact or another ordering constraint on later conflicting reservations is insufficient for authorization safety plus first-arrival liveness;
      \item the fair reserve/use protocol satisfies authorization safety and first-arrival liveness.
  \end{itemize} 
\end{theorem}

\begin{proof}
  Clause~(i) is Lemma~\ref{lem:fairnecessary}, clause~(ii) is Lemma~\ref{lem:bindingnecessary}, and clause~(iii) is Proposition~\ref{prop:reserveuse}.
\end{proof}

\section{Discussion}

\paragraph*{Why this is not just revalidation.}
If validators can compute the current state from the log, they can revalidate each command against $S_j$, and binding validity is automatic. That is not our model. Here the policy state is committed but hidden, so old evidence is stale information and the prefix contains no conflict fact unless the protocol writes one. Binding is therefore a consensus-level obligation.

\paragraph*{Why this is not just locking.}
Under a trusted lock manager or trusted FIFO queue, admission and execution are atomic outside the Byzantine log. In Byzantine SMR, the log is the only shared durable medium. A lock that is not in the log can be ignored or reordered by the leader, a log entry that is not binding is only an attestation. The lock must be both consensus-visible and fairly ordered.

\paragraph*{Why this is not just fair ordering.}
Order-fairness constrains the relative order of commands, not the monotonicity of an application's validity predicate. Lemma~\ref{lem:bindingnecessary} gives the separation: even perfectly fair ordering of historical claims leaves later validity unchanged unless the first claim creates a public binding fact.

\paragraph*{Beyond one resource.}
We do not claim a free extension to replenishing balances, multiple units, multi-resource transactions, or partial fair orders. For example, if $T_1$ and $T_2$ spend Alice's balance while an intervening $T_3$ replenishes it, $T_1\to T_3\to T_2$ may be legal: the conflict is version-dependent and can disappear after $T_3$. Such an execution requires fresh authorization rather than permanent exclusion and lies outside our one-shot theorem. With $k$ units and $m>k$ contenders, one must additionally agree on a winner set, while partial fair orders may contain cycles.

\paragraph*{Relation to prior work.}
Fair-ordering protocols~\cite{kelkar2020order,kelkar2023themis,zhang2020pompe} repair the arrival-order-vs-commit-order gap but do not by themselves specify what a committed value means for later application validity. Validated Byzantine agreement~\cite{cachin2001broadcast} covers public predicates, our setting is different because current authorization is committed but not reconstructible by validators. The attestation/reservation distinction is the optimistic-vs-pessimistic concurrency-control boundary~\cite{kung1981optimistic} transported into Byzantine SMR with hidden state. Dynamic-concurrency work makes conflict explicitly state-dependent~\cite{kuznetsov2025conflicts}, our relation $\#_r$ is its non-replenishing, one-shot specialization.

\bibliography{lipics-v2021-sample-article}

\end{document}